\documentclass[runningheads,a4paper]{llncs}
\usepackage{amsmath}
\usepackage{float}
\usepackage{amssymb}
\usepackage{graphicx}
\usepackage{subfigure}
\usepackage{url}
\usepackage{graphicx}
\usepackage{caption}

\urldef{\mailsa}\path|qinxl@casit.ac.cn,|
\urldef{\mailsb}\path|{yongfeng, wuwenyuan}@cigits.ac.cn,|
\urldef{\mailsc}\path|reid@uwo.ca|
\newcommand{\keywords}[1]{\par\addvspace\baselineskip
\noindent\keywordname\enspace\ignorespaces#1}

\begin{document}

\mainmatter  % start of an individual contribution

\title{Structural analysis of high-index DAE for process
simulation}

%\titlerunning{Lecture Notes in Computer Science: Authors' Instructions}

\author{Xiaolin Qin$^{1,4}$%
\thanks{Corresponding author.}%
\and Wenyuan Wu$^{2}$\and Yong Feng$^{1}$\and Greg Reid$^{1,3}$}
\authorrunning{Xiaolin Qin et al.}

\institute{$^1$Chengdu Institute of Computer Applications, CAS,
Chengdu 610041, China\\
$^{2}$Chongqing Institutes of Green and Intelligent Technology, CAS,
Chongqing 401122, China\\
$^{3}$Applied Mathematics Department, University of Western Ontario,
London, N6A 5B7, Canada\\
$^{4}$Department of Mathematics, Sichuan University, Chengdu 610064, China\\
 \mailsa \mailsb \mailsc }


\toctitle{Structural analysis of high-index DAE for process
simulation} \tocauthor{Xiaolin Qin et al.} \maketitle

\begin{abstract}
This paper deals with the structural analysis problem of dynamic
lumped process high-index DAE models. We consider two methods for
index reduction of such models by differentiation: Pryce's method
and the symbolic differential elimination algorithm rifsimp.
Discussion and comparison of these methods are given via a class of
fundamental process simulation examples. In particular, the
efficiency of the Pryce method is illustrated as a function of the
number of tanks in process design.

\end{abstract}

\section{Introduction}
Differential-algebraic equations (DAE) systems arise naturally when
modelling many dynamic systems. Dynamic process models and their
properties form the background of any process control activity
including model analysis, model parameter and structure estimation,
diagnosis, regulation or optimal control. In particular, the
structural analysis of dynamic lumped process models forms an
important step in the model building procedure \cite{book1}, and it
is used for the determination of the solvability properties of the
model. Furthermore, the dynamic lumped process models often require
the consistent initial conditions and solution of high-index
differential-algebraic systems.

The index is a notion used in the theory of DAEs for measuring the
distance from a DAE to its related ODE. High-index DAE systems need
prolongation (differentiation) to reveal all the system's
constraints, and to determine consistent initial conditions. The key
steps include identifying all hidden constraints on formal power
series solutions in the neighborhood of a given point, and are
required to prepare the system for numerical integration. So for
such differential systems, prolongation is unavoidable. In the
present work, Pryce developed a Taylor series method based on his
structural analysis method \cite{jour1,jour2} and on Pantelides'
work in \cite{jour3}. Pantelides' method gives a systematic way to
reduce high-index systems of differential-algebraic equations to
lower index, by selectively adding differentiated forms of the
equations already present in the system. It is implemented in
several significant equation-based simulation programs such as
gPROMS \cite{jour10}, Modelica \cite{book2} and EMSO \cite{jour11}.
However, the algorithm can fail in some instances. Pryce's
structural analysis is based on solving an assignment problem, which
can be formulated as an integer linear programming problem. It finds
all the constraints for a large class of ODE using only
prolongation, which can be considered as fast prolongation method.
Corless et al. show Pryce's method can be extended to give a
polynomial cost method for numerical solution of differential
algebraic equations \cite{jour4}. Wu et al. give a differential
algebraic interpretation of Pryce's method for ODE, which
generalizes to a certain class of PDE for finding missing
constraints \cite{proceeding1}. Mani shows how pre-symbolic
simplification can usefully extend the applicability of the Pryce
method on models produced by MapleSim \cite{jour12}.

In \cite{jour6,jour7}, Leitold et al. propose the structural
analysis of process models using their representation graphs for the
determination of the most important solvability property of lumped
dynamic models: the differential index. Their graph-theoretical
method depends on the change in the relative position of
underspecified and overspecified subgraphs and has an effect to the
value of the differential index for complex models. If these
subgraphs move further from their original positions the value of
differential index increases. In this paper, we consider other
approaches for the structural analysis of dynamic lumped process
models for high-index DAE systems. In particular, we consider
Pryce's method and the symbolic differential elimination package
rifsimp. Pryce's method is a robust and reliable method for
remedying the drawback of the approach \cite{jour6,jour7} and doing
so automatically. This is a powerful way to determine the index of
the system, its number of degrees of freedom, and exactly which
components should be given initial values. The key idea is taken
from Pryce's signature-method. The nice feature of the work is a
simple and straightforward method for analysing the structure of a
differential algebraic system.

The rest of this paper is organized as follows. Section 2 describes
Pryce's method and introduces the symbolic differential elimination
package rifsimp in Maple. Section 3 gives the structural analysis of
simple process models using these approaches. Section 4 gives some
experimental results. The final section concludes this paper.

\section{Preliminaries}
In this section, we give a brief review of Pryce's method and some
remarks, and present the symbolic differential elimination with
Maple's rifsimp package.
\subsection{Pryce's method}
We review below the main steps of Pryce's structural analysis and
the corresponding algorithm following \cite{jour1,jour2}. We
consider an input system of $n$ equations $f=0$, where $f=(f_1, f_2
\cdots, f_n)$ in $n$ dependent variables $x_1(t), x_2(t), \cdots,
x_n(t)$.

Step 1. Form the $n \times n$ signature matrix $\Sigma =
\sigma_{ij}$ of the DAE, where
\begin{center}
${\sigma_{ij}}=\begin{cases}
 highest \ order\ of \ derivative \ to \ which \ the\ variable \\
  x_j\ appears\ in\ equation\ f_i, \\
or \ -\infty\ if \ the\ variable\ does\ not\ occur.
\end{cases}$
\end{center}

Step 2. Solve an assignment problem to find a HVT $(highest\ value\
transversal)$, which is a subset of indices $(i, j)$ describing just
one element in each row and each column, such that $\sum
\sigma_{ij}$ is maximized and finite.

Step 3. Determine the offsets of the problem, which are the vectors
$\mathbf{c}= (c_i)_{1\leq i\leq n}, \mathbf{d}=(d_j)_{1\leq j\leq
n}$, the smallest such that $d_j - c_i \geq \sigma_{ij}$, for all
$1\leq i\leq n, 1\leq j\leq n$ with equality on the HVT. This
problem can be formulated as an integer linear programming problem
(LPP) in the variables $\mathbf{c} = (c_1, c_2, \cdots, c_n$) and $
\mathbf{d} = (d_1, d_2, \cdots, d_n)$:
\begin{subequations}
\begin{eqnarray}
&&Minimize \ \ z = \sum_j d_j - \sum_i c_i, \label{equ:a}\\
&&where \ \ d_j - c_i \geq \sigma_{ij} \ for\ all\ (i, j), \label{equ:b}\\
&&c_i \geq 0 \ for\ all\ i.\label{equ:c}
\end{eqnarray}
\end{subequations}

The structural index is then defined as
\begin{center}
$\nu=\max_{i}c_i+\begin{cases}
 0 \ for\ all\ d_j>0 \\
1 \ for\ some\ d_j=0.
\end{cases}$
\end{center}
The structural index is no less that the differential index on first
order DAE.

Step 4. Form the $n\times n$ system Jacobian matrix $\mathbf{J}$
where
\begin{center}
$\mathbf{J}_{ij}=\begin{cases}
 \frac{\partial f_i}{\partial ((d_j-c_i)th\ derivative\ of\ x_j)} \ \ if\ this\ derivative\ is\ present\ in\ f_i  \\
0 \ \ \ \ \ \ \ \ \ \ \ \ \ \ \ \ \ \ \ \  \ \ \ \ \ \ \ \ \ \ \ \ \
\  otherwise.
\end{cases}$
\end{center}

Step 5. Choose a consistent point. If $\mathbf{J}$ is non-singular
at that point, then the solution can be computed with Taylor series
or numerical homotopy continuation techniques in a neighborhood of
that point.
\paragraph{Remark 1.}
The computation of $\mathbf{c}$ and $\mathbf{d}$ only involves the
information on differential order and is consequently very fast in
Step 3 of Pryce's method. This problem is dual to the assignment
problem. The time complexity of the assignment problem can be done
at polynomial cost by using the Hungarian Method \cite{jour8}.
Generally, such problems can be solved very efficiently in practice.
\paragraph{Remark 2.} After we obtain the number of prolongation
steps $c_i$ for each equation from Step 3 of Pryce's method, we can
enlarge the system of equations using $\mathbf{c}$. We assume $c_1
\geq c_2 \geq \cdots \geq c_n$, and let $k_c =\max_i c_i= c_1$,
which is closely related to the index of DAEs. Consider the
equations obtained by taking the $t$-derivative of $f_1^{(0)}=f_1=0$
up to the $c_i$th derivative, $1\leq i \leq n$, that is the
collection
\begin{equation}
 \label{eq:1}
  \left\{ \begin{aligned}
f_1^{(0)}, &f_1^{(1)},& \cdots,& f_1^{(c_1)}  \\
& \vdots && \\
 f_n^{(0)},& f_n^{(1)},& \cdots,&  f_n^{(c_n)}
  \end{aligned} \right\} = 0,
                           \end{equation}
where $^{(l)}$ denotes $d^l/dt^l$ \footnote{$^{(l)}$ is defined by
the same way for the rest of this paper.}. By the definition of
$\sigma_{ij}$ and inequalities (1b), the derivatives of the $x_j$
that occur in equations (2) all lie in this set:
\begin{equation}
\centering \left\{ \begin{aligned}
x_1^{(0)}, &x_1^{(1)}, &\cdots, &x_1^{(d_1)}, \\
&\vdots&& \\
x_n^{(0)}, &x_n^{(1)}, &\cdots, &x_n^{(d_n)}.
\end{aligned}\right.
\end{equation}
Represent (3) as a vector X, then (2) can be written as a system
\begin{equation}\label{eq:triple}
0=F(t, X)=\begin{pmatrix}
 F_{0} (t, X_{0})\\
 F_{1} (t, X_{0}, X_{1})\\
 \vdots \\
 F_{k_d} (t, X_0, X_1, \cdots, X_{k_d-1}, X_{k_d})
 \end{pmatrix},
\end{equation}
where $k_d =\max_j d_j = d_1$, and assume $d_1 \geq d_2 \geq \cdots
\geq d_n$. In particular, for $0\leq i \leq k_c$, $F_i$ has fewer
variables than $F_{i+1}$. The block structure form $B_i (0\leq i <
k_c)$ in the case $c_i=c_{i+1}+1$ is given in Table 1.
\begin{table}[H]
\caption{The triangular block structure of $F$ for the case $c_i= c_{i+1}+1$} % title of Table
\centering % used for centering table
\begin{tabular}{|c |c| c| c| c|} % centered columns (4 columns)
\hline%inserts horizontal lines
$B_0$&$B_1$ &$\cdots$ & $B_{k_c-1}$&$B_{k_c}$\\[0.5ex]
\hline % inserts single horizontal line
$F_1^{(0)}$&$F_1^{(1)}$ &$\cdots $&$F_1^{(c_1-1)}$&$F_1^{(c_1)}$\\
&$F_2^{(0)}$&$\cdots$ &$F_2^{(c_2-1)}$&$F_2^{(c_2)}$\\
& &$\vdots$ &$\vdots$ &$\vdots$ \\
&&$F_n^{(0)}$&$\cdots$&$F_n^{(c_n)}$\\
\hline %inserts single line
\end{tabular}
\label{table:nonlin} % is used to refer this table in the text
\end{table}
\paragraph{Remark 3.}
Fast prolongation produces a simplified system to which a standard
numeric solver can be efficiently applied.

\subsection{Symbolic differential elimination}

Maple's rifsimp package can be used to simplify small- and
middle-scale DAEs, and overdetermined systems of polynomially
nonlinear PDEs or ODEs and inequations to a more useful form
\cite{jour9}. For the DAEs and ODEs the only independent variable is
time. It processes systems of polynomially nonlinear PDEs with
dependent variables $u_1, u_2, \cdots, u_n$, which can be functions
of several independent variables.

The key idea of algorithm is substitution and differential
elimination, which requires a ranking to be defined on the dependent
variables and their derivatives. A basic step of differential
elimination algorithms linearly appearing is to write the system in
solved form with respect to each highest ranked derivative. It is
treated by methods involving a combination of Gr$\ddot{o}$ber bases
and Triangular decompositions. Another key step in such algorithms
is the taking of integrability conditions between equations.

The rifsimp algorithm is essentially an extension of the Gaussian
elimination to DAEs and systems of nonlinear PDEs. It differentiates
the leading nonlinear equations and then reduce them with respect to
the leading linear equations. If zero is obtained, it means that the
equation is a consequence of the leading linear equations. If not,
it means that this equation is a new constraint to the system. This
is repeated until no new constraints are found. See Section 3.3 for
a simple example.

\section{Structural analysis of simple process models}

In this section, we apply the above techniques to structure analysis
of dynamic lumped process models DAE systems. The model is taken
from dynamic process simulation and multi-domain modeling and
simulation of complex systems. Here, the cascade of perfectly
stirred tank reactors yields the basic examples of the paper, see
Fig.1.
\begin{figure}
\centering
\includegraphics[height=4.2cm]{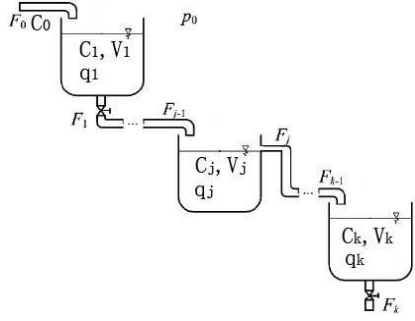}
\caption{Sequence of liquid tanks, where for the $i$-th tank
$F_{i-1}$ and $F_i (i= 1, 2, \cdots, k$) are the inlet and outlet
flow rate, $C_i$ is the concentration and  $V_i$ is the fixed volume
of the $i$-th tank.} \label{fig:example}
\end{figure}
\vspace*{-10pt}
\subsection{Main algorithm}
Suppose a system consists of $k$ perfectly stirred tank reactor. A
feed of concentration $C_0(t)$ is fed into the first tank. The
concentrations in the tanks are described by the following equation:
\begin{equation}\label{eq:orig}
C_{i}^{(1)}=\frac{q(t)}{V_i(t)}(C_{i-1}(t)-C_i(t))\ \   i=1, 2,
\cdots, k
\end{equation}
where $C_i(t)$ is the concentration in the tank $i$, $q(t)$ is the
flow rates from tank to tank and $V_i(t)$ is the fixed volume of the
tank $i$. Thus the flow rates between the tanks $q_{i}(t)$ are all
the same that $q_{i}(t) = q(t) =Q(t)$, where $Q(t)$ is a specified
function of $t$.

In general, there are two different specifications that can be added
to
these equations according to the modelling goal: \\
a) in dynamic simulation studies the feed concentration $C_0(t)$ is
given by $C_0=C0(t)$;\\
b) in dynamic design the product concentration $C_k(t)$ is given by
$C_k(t)=Ck(t)$.

When applied to process system a) and b), the main steps of our
approach are:

Step 1. Construct the original system as follows based on the
equation (\ref{eq:orig}):
\begin{equation}\label{eq:orisys}
\begin{split}
F := [C_i(t)^{(1)} = q(t)(C_{i-1}(t)-C_i(t))/V_i(t), i = 1, 2,
\cdots, k, \\ V_i(t)^{(1)} = 0,  i = 1, 2, \cdots, k, q(t) = Q(t)],
\end{split}
\end{equation}
where $k$ is the number of tanks.

Step 2. Obtain the original condition and add it into $F$. There are two cases:\\
a) in dynamic simulation the tank feed concentration $C_0(t)$ is
given as a function of time then get $2k + 1$ equations in $2k + 1$
unknowns: $C_0(t) = C0(t)$, this is essentially index 1 no matter
what $k$ is, and is a trivial system. Symbolic differential
elimination can be used for case a);\\
b) in dynamic design the product concentration $C_k(t)$ is given as
a function: $C_k(t) = Ck(t)$. It is a nontrivial system, which is
high-index as $k$ increased.

Step 3. Call the Pryce's algorithm of Section 2 to solve the vector
$\mathbf{c}$ and $ \mathbf{d}$, and enlarge the initial system by
fast prolongation. Alternatively symbolic differential elimination
can be used for case b).

Step 4. Check the Jacobian matrix $\mathbf{J}$ with the coefficients
of highest derivatives equations and compute the consistent point.
\paragraph{Remark 4.}Based on the structure analysis of Pryce's method, it is practical
and efficient for dynamic lumped process models DAE systems. In
general, the goal of structural analysis of DAEs is to differentiate
the equations in such a way that the coefficient (Jacobian) matrix
of the highest derivatives is non-singular. It means that some
equations need prolongations on independent variable to balance the
coefficients matrix. So it can computing Jacobian matrix of the
lower derivatives equations by an iterative procedure for finding
all consistent points.
\paragraph{Remark 5.} Compared with the structural analysis of process models using their representation graphs method,
the advantages of our algorithm are:\\
$\bullet$ We efficiently apply the fast numerical and symbolic
computations to a wide variety of
physical models generated by the equation-based technique.\\
$\bullet$ For the large models, we can keep the structural index of
system remaining unchanged. Moreover, the prolongation system has a
favorable block triangular structure to compute the missing initial
conditions more efficiently.

\subsection{Main results}
For the general dynamic lumped process models DAE systems, we can
obtain the offsets of vector
$$\mathbf{c} = (0, 1, 2, \cdots, k-1,
0, 0, 1, 2, \cdots, k)$$
 and
  $$\mathbf{d} = (0, 1, 2, \cdots, k, 1,
1, 2, \cdots, k-1, k-1)$$ by Pryce's method.
Therefore, we have the
following ranking of dependent variables.
\[
\left( {\begin{array}{c}
 C_k
 \end{array} } \right)
 \prec \left( {\begin{array}{c}
 C_{k-1} \\
C_k^{(1)}\\
V_{k-1}\\
V_{k}\\
q
 \end{array} } \right)
\prec \left( {\begin{array}{c}
 C_{k-2} \\
C_{k-1}^{(1)}\\
C_k^{(2)}\\
V_{k-1}^{(1)}\\
V_k^{(1)}\\
q^{(1)}
 \end{array} } \right)
\prec \cdots \prec \left( {\begin{array}{c}
 \\
C_1\\
C_2^{(1)}\\
\vdots\\
C_k^{(k-1)}\\
V_1\\
V_2\\
\vdots\\
V_{k-1}^{(k-3)}\\
V_k^{(k-2)}\\
q^{(k-2)}
 \end{array} } \right)
\prec \left( {\begin{array}{c}
 C_0\\
C_1^{(1)}\\
C_2^{(2)}\\
\vdots\\
C_k^{(k)}\\
V_1^{(1)}\\
V_2^{(1)}\\
\vdots\\
V_{k-1}^{(k-2)}\\
V_k^{(k-1)}\\
q^{(k-1)}
 \end{array} } \right)
\]

Based on the above ranking, we can obtain the sequence of solving
initial value problem for the dynamic lumped process models DAE
systems. It is equivalently the block-triangular system that has
full row rank for each $k$.

From (\ref{eq:orisys}), $F$ and the original condition
$C_k(t)=Ck(t)$ have
$$ M = (\sum
c_i) + (2k+2) = ((\sum_1^{k-1} i) + (\sum_{1}^{k} i)) + (2k + 2) =
k^2+2k+2$$ components. The number of variables is $$ N = (\sum d_j)
+ (2k + 2) = ((\sum_1^{k} j) + 1 + (\sum_1^{k-1} j) + (k-1))+ (2k +
2) = k^2+3k+2.
$$

Considered as $M$ algebraic equations in $N$ variables, it has a
solution $(t^{*}, X^{*})$, where
$$X =(C_0, C_1^{(1)}, C_2^{(2)},
\cdots, C_k^{(k)}, V_1^{(1)}, V_2^{(1)}, \cdots, V_{k-1}^{(k-2)},
V_k^{(k-1)}, q^{(k-1)});$$ and that $\mathbf{J}$ is non-singular.
Therefore, $(t^{*}, X^{*})$ is a consistent point. In a neighborhood
of the point $(t^{*}, X^{*})$, the solution manifold has $D$ degrees
of freedom \cite{jour2}.
\begin{lemma}
At a point $(t^{*}, X^{*})$ in $\emph{M}$ where $\mathbf{J}$ is
non-singular, $\emph{M}$ is locally a manifold of dimension $\pi +
1$ parameterized. The solution manifold has $D$ degrees of freedom,
where
$$D=\pi= \sum d_j - \sum c_i = N - M.$$
\end{lemma}
The above shows that if we find a solution $(t^{*}, X^{*})$, this is
a consistent point, and if the number of degrees of freedom $D > 0$
there are other consistent points nearby for the same $t$.
\begin{theorem}
The general dynamic lumped process models DAE systems have degrees
of freedom $D = \sum d_j- \sum c_i = k$, where k is the number of
tanks. The structural index is $k+1$.
\end{theorem}
\begin{proof}
From Lemma 1, we have degrees of freedom
\begin{equation}\label{eq:dof}
\begin{split}
D = \sum d_j -\sum c_i =
(\sum_{1}^{k}j)+1+(\sum_{1}^{k-1}j)+k-1-\\((\sum_{1}^{k-1}j)+(\sum_{1}^{k-2}j)+k-1+k)
= k+1+k-1+0-k = k.
\end{split}
\end{equation}
Because the $d_1=0$, the structural index is $\max_i{(c_i)}+1=k+1$.
\end{proof}
Here, we give the degrees of freedom and structural index of the
general dynamic lumped process models DAE systems that is a function
of the number of tanks $k$.
\subsection{A detailed example}
\newcounter{num1}
\begin{example}We propose a simple example to set the number of tanks $k:=3$ case
a) to illustrate the rifsimp algorithm.
\begin{list}{Step \arabic{num1}:}{\usecounter{num1}\setlength{\rightmargin}{\leftmargin}}
\item Construct the original system as follows:
\begin{eqnarray*}sys := [C_1(t)^{(1)} =
\frac{q(t)*(C_0(t)-C_1(t))}{V_1(t)}, C_2(t)^{(1)}  =
\frac{q(t)*(C_1(t)-C_2(t))}{V_2(t)},\\ C_3(t)^{(1)}  =
\frac{q(t)*(C_2(t)-C_3(t))}{V_3(t)}, V_1(t)^{(1)}  = 0, V_2(t)^{(1)}
= 0, V_3(t)^{(1)}  = 0, q(t) = Q(t)];
\end{eqnarray*}
\item Obtain the original condition $C_0(t) = C0(t)$, and add it to $sys$;
\item Call rifsimp algorithm to reduce the system as follows:
\begin{eqnarray*}
[C_1(t)^{(1)}  = \frac{Q(t)*C0(t)-Q(t)*C_1(t)}{V_1(t)},
C_2(t)^{(1)}  = \frac{Q(t)*C_1(t)-Q(t)*C_2(t))}{V_2(t)},\\
C_3(t)^{(1)}  = \frac{Q(t)*C_2(t)-Q(t)*C_3(t)}{V_3(t)}, V_1(t)^{(1)}
= 0, V_2(t)^{(1)}  = 0,\\ V_3(t)^{(1)}  = 0, C_0(t) = C0(t),
 q(t) = Q(t), V_1(t) \neq 0, V_2(t) \neq 0, V_3(t) \neq 0].
\end{eqnarray*}
\end{list}
\end{example}
\paragraph{Remark 6.} In this paper,
we consider the modelling goal for case a) by the rifsimp algorithm.
The main reason is the specific structure of models, which is the
quasi-triangular system and has $C_0(t) = C0(t)$ specified.
Therefore, it is only simple check. But it becomes rapidly more
complicated as the number $k$ increased for case b).
\newcounter{num2}
\begin{example}We propose a simple example to set the number of tanks $k:=4$ case
b) and illustrate our algorithms.
\begin{list}{Step \arabic{num2}:}{\usecounter{num2}\setlength{\rightmargin}{\leftmargin}}
\item Construct the original system as follows:
\begin{eqnarray*}sys := [D_1=C_1(t)^{(1)} -
\frac{q(t)*(C_0(t)-C_1(t))}{V_1(t)}=0, D_2=C_2(t)^{(1)}
\\-\frac{q(t)*(C_1(t)-C_2(t))}{V_2(t)}=0, D_3= C_3(t)^{(1)} -
\frac{q(t)*(C_2(t)-C_3(t))}{V_3(t)}=0, D_4=C_4(t)^{(1)} \\ -
\frac{q(t)*(C_3(t)-C_4(t))}{V_4(t)}=0, D_5=V_1(t)^{(1)}  = 0,
D_6=V_2(t)^{(1)}  = 0,\\ D_7=V_3(t)^{(1)}  = 0,D_8=V_4(t)^{(1)}  =
0, D_9=q(t) - Q(t)=0];
\end{eqnarray*}
\item Obtain the original condition $C_4(t) = C4(t)$, and add $D_{10}=C_4(t) - C4(t)=0$ to $sys$;
\item Obtain the variables list $\mathit{variables} := [{C_{0}},
\,{C_{1}}, \,{C_{2}}, \,{C_{3}}, \,{C_{4}}, \,{V_{1}}, \,{V_{2}},
\,{V_{3}}, \,{V_{4}}, \,q]$;
\item Call the Pryce's method and solving this integer LPP by LPSolve
in the Optimization package of Maple, we obtain the fast
prolongation times for the $i$-th equation from $\mathbf{c}$, and
the highest order of derivative variables
from $\mathbf{d}$ as follows:\\
$c_1=0, c_2=1, c_3=2, c_4=3, c_5=0, c_6=0, c_7=1, c_8=2, c_9=3, c_{10}=4,\\
d_1=0, d_2=1, d_3=2, d_4=3, d_5=4, d_6=1, d_7=1, d_8=2, d_9=3,
d_{10}=3$. \\
Therefore, according to the $c_i$ values it can be prolonged for the
corresponding equations automatically. Enlarged sets of variables:\\
$\{C_0; C_1, C_1^{(1)}; C_2, C_2^{(1)}, C_2^{(2)}; C_3, C_3^{(1)},
C_3^{(2)}, C_3^{(3)}; C_4, C_4^{(1)}, C_4^{(2)}, C_4^{(3)},
C_4^{(4)};\\ V_1, V_1^{(1)}; V_2, V_2^{(1)}; V_3, V_3^{(1)},
V_3^{(2)}; V_4, V_4^{(1)}, V_4^{(2)}, V_4^{(3)}; q, q^{(1)},
q^{(2)}, q^{(3)}\},$\\ equations: \\$\{D_1; D_2, D_2^{(1)}; D_3,
D_3^{(1)}, D_3^{(2)}; D_4, D_4^{(1)}, D_4^{(2)}, D_4^{(3)}; D_5;
D_6; D_7, D_7^{(1)}; D_8, D_8^{(1)}, D_8^{(2)};\\ D_9, D_9^{(1)},
D_9^{(2)}, D_9^{(3)}; D_{10}, D_{10}^{(1)}, D_{10}^{(2)},
D_{10}^{(3)}, D_{10}^{(4)}\}.$ The system Jacobian $\mathbf{J}$ is:
\[
 \mathbf{J} :=  \left[ {\begin{array}{c}
 - {\displaystyle \frac {\mathrm{q}(t)}{{V_{1}}(t)}} \, \ \ \ \ \ \ \,1\, \
 \ \ \ \ \
\,0\, \ \ \ \ \ \,0\,\ \ \  \ \,\ 0\, \,\ \ \ \ \ 0\,\ \ \ \ \
\,0\,\
\ \ \ \ \ \,0\, \ \ \ \ \ \ \ \ \,0\, \ \ \ \ \ \ \,0 \\
 [2ex] 0\, \, \ -
{\displaystyle \frac {\mathrm{q}(t)}{{V_{2}}(t)}} \, \ \  \,1\, \ \
\ \,0\, \ \ \ \,0\, \ \ \ \ \,0\, \ \ \ \,{\displaystyle \frac
{\mathrm{q}(t) \,({C_{1}}(t) - {C_{2}}(t))}{{V_{2}}(t)^{2}}} \, \ \
\
   0\, \ \ \  \,0\, \ \ \ \ \ \,0
\\
[2ex] 0\, \ \ \ \,0\, \ \, - {\displaystyle \frac
{\mathrm{q}(t)}{{V_{3}}(t)} } \, \ \ \  \,1\, \ \ \,0\, \ \ \ \
\,0\, \ \ \ \,0\, \ \ \,{\displaystyle \frac {
\mathrm{q}(t)\,({C_{2}}(t) -
{C_{3}}(t))}{{V_{3}}(t)^{2}}}\, \ \ \ \ \, 0\, \ \ \ \ \,0 \\
[2ex] 0\, \ \,0\, \ \,0\, \ \, - {\displaystyle \frac
{\mathrm{q}(t)}{{V_{ 4}}(t)}} \, \,1\, \ \,0\, \ \,0\, \ \,0\, \
\,{\displaystyle \frac { \mathrm{q}(t)\,({C_{3}}(t) -
{C_{4}}(t))}{{V_{4}}(t)^{2}}} \, \,
 - {\displaystyle \frac {{C_{3}}(t) - {C_{4}}(t)}{{V_{4}}(t)}}
 \\ [2ex]
0\, \ \ \ \ \ \,0\, \ \ \ \ \ \ \ \,0\, \ \ \ \ \ \ \,0\, \ \ \ \ \
\ \ \,0\, \ \ \ \ \ \ \,1\, \ \ \ \ \ \ \ \ \,0\, \ \ \ \ \ \ \,0\,
\ \ \ \ \ \ \ \ \,0\, \ \ \ \ \ \ \ \,0
 \\
0\, \ \ \ \ \ \,0\, \ \ \ \ \ \ \ \,0\, \ \ \ \ \ \ \,0\, \ \ \ \ \
\ \ \,0\, \ \ \ \ \ \ \,0\, \ \ \ \ \ \ \ \ \,1\, \ \ \ \ \ \ \,0\,
\ \ \ \ \ \ \ \ \,0\, \ \ \ \ \ \ \ \,0
 \\
0\, \ \ \ \ \ \,0\, \ \ \ \ \ \ \ \,0\, \ \ \ \ \ \ \,0\, \ \ \ \ \
\ \ \,0\, \ \ \ \ \ \ \,0\, \ \ \ \ \ \ \ \ \,0\, \ \ \ \ \ \ \,1\,
\ \ \ \ \ \ \ \ \,0\, \ \ \ \ \ \ \ \,0
 \\
0\, \ \ \ \ \ \,0\, \ \ \ \ \ \ \ \,0\, \ \ \ \ \ \ \,0\, \ \ \ \ \
\ \ \,0\, \ \ \ \ \ \ \,0\, \ \ \ \ \ \ \ \ \,0\, \ \ \ \ \ \ \,0\,
\ \ \ \ \ \ \ \ \,1\, \ \ \ \ \ \ \ \,0
 \\
0\, \ \ \ \ \ \,0\, \ \ \ \ \ \ \ \,0\, \ \ \ \ \ \ \,0\, \ \ \ \ \
\ \ \,0\, \ \ \ \ \ \ \,0\, \ \ \ \ \ \ \ \ \,0\, \ \ \ \ \ \ \,0\,
\ \ \ \ \ \ \ \ \,0\, \ \ \ \ \ \ \ \,1
 \\
0\, \ \ \ \ \ \,0\, \ \ \ \ \ \ \ \,0\, \ \ \ \ \ \ \,0\, \ \ \ \ \
\ \ \,1\, \ \ \ \ \ \ \,0\, \ \ \ \ \ \ \ \ \,0\, \ \ \ \ \ \ \,0\,
\ \ \ \ \ \ \ \ \,0\, \ \ \ \ \ \ \ \,0
\end{array}}
 \right]
\]
\item Computing the Jacobian matrix $\mathbf{J}=
-\frac{q(t)^4}{V_1(t)V_2(t)V_3(t)V_4(t)}$, which is non-singular.
And then we can compute the consistent point by numerical methods,
such as Taylor series methods, Homotopy methods.
\end{list}
\end{example}
\paragraph{Remark 7.} In particular, we can obtain the coefficient
(Jacobian) matrix that is sparse dramatically. For the highest
derivatives, the determinant of Jacobian matrix is  $det \
J=-\frac{q(t)^k}{V_1(t)V_2(t)\cdots V_k(t)}$ where $k$ is the number
of tanks.
\section{Experimental Results}

An efficient practical implementation of Pryce's method is in
\emph{Maple}. The following examples run in the platform of Maple
and Inter(R) Core(TM) i3 2.40GHz, 2.00G RAM. We give some
experimental results using symbolic differential elimination and
fast prolongation for structural analysis of dynamic lumped process
models DAE systems. In Fig. $2$, we present the time for symbolic
differential elimination by rifsimp of Maple package and fast
prolongation as the number of tank reactors $k$ increased. In Fig.
$3$, we present the memory usage for symbolic differential
elimination by rifsimp of Maple package and fast prolongation as the
number of tank reactors $k$ increased.

The system Jacobian is very sparse for case b). Its determinants are
evaluated symbolically to be $det \
J=-\frac{q(t)^k}{V_1(t)V_2(t)\cdots V_k(t)}$ where $k$ is the number
of tanks, and is non zero. In other examples the alternative is to
usually find an approximate point satisfying the constraints by
numerical method (eg. Homotopy method) and evaluate the condition
number of the Jacobian to carry out he validation.
 \vspace*{-10pt}
\begin{figure}[H]
\centering
\includegraphics[height=6.2cm]{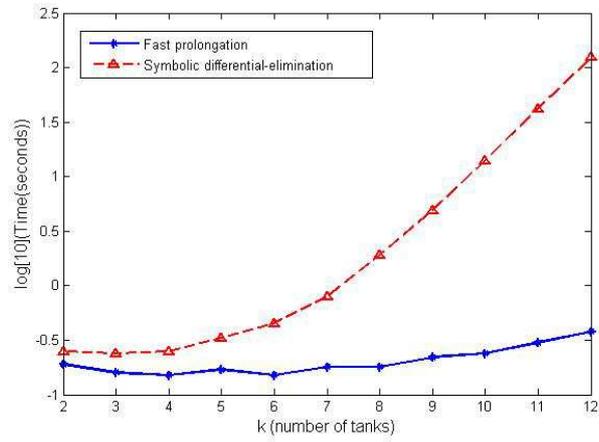}
\caption{Time for structural analysis of dynamic lumped process
models DAE systems using symbolic differential elimination and fast
prolongation.} \label{fig:example}
\end{figure}
\vspace*{-10pt}
\begin{figure}[H]
\centering
\includegraphics[height=6.2cm]{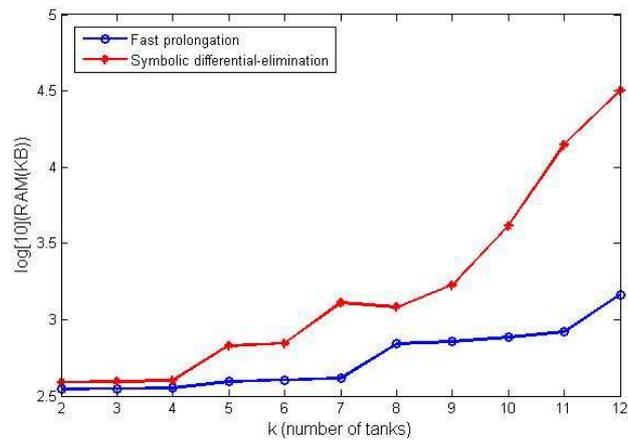}
\caption{Memory usage for structural analysis of dynamic lumped
process models DAE systems using symbolic differential elimination
and fast prolongation.} \label{fig:example}
\end{figure}
\vspace*{-10pt}

From Figures $2$ and $3$:\\
$\bullet$ The time of structural analysis of dynamic lumped process
models DAE systems for fast prolongation is small and ultimately
grows slowly in the range of degrees of freedom considered. The time
for symbolic differential elimination method grows much faster. The
main reason is that fast prolongation only needs to solve an integer
linear programming problem, but the symbolic differential
elimination needs a large number of eliminations and differentiates.
Therefore, the symbolic differential elimination is more difficult
for the general high-index DAE
systems. \\
$\bullet$ The memory shows steady growth as the number $k$
increases. The memory usage of symbolic differential elimination
grows very quickly.

The above analysis and experimental results, motivates consideration
of hybrid techniques involving a combination of symbolic
differential elimination and fast prolongation for large DAE models.
However, symbolic computations have the disadvantage of intermediate
expression swell. In the future, we would like to consider a
combination of partial symbolic differential elimination and fast
prolongation to model and simulate realistic physical models. We
hope to give the specific structural analysis algorithms that
exploit the form of systems appearing in applications.

\section{Conclusion}
In this paper, we have investigated the high-index structural
analysis problem for the class of dynamic lumped process models DAE
systems by Pryce's method and symbolic differential elimination. We
designed the algorithm to automatically analysis the structural of
simple process models, and showed that the rifsimp algorithm of
Maple package reduces the original system to standard form. We also
gave the degrees of freedom and structural index of the dynamic
lumped process models DAE systems that is a function of the number
of tanks $k$. Moreover, those approached can be generalized to a
wide variety of physical models and analyzed the structural of
square and non-square nonlinear DAE and PDAE systems.
\subsubsection*{Acknowledgments.} This work was partially supported by the National
Basic Research Program of China (2011CB302402), the West Light
Foundation of the Chinese Academy of Sciences, the National Natural
Science Foundation of China (Grant NO. 91118001, 11171053), and the
Chinese Academy of Sciences Visiting Professorship for Senior
International Scientists(Grant No.2010T2G31).


\begin{thebibliography}{20}

\bibitem{book1} Hangos, K.M. and Cameron, I.T.: Process Modelling and Model Analysis,
Academic Press, London (2001)

\bibitem{jour1} Pryce, J.D.: Solving high-index DAEs by Taylor Series. Numerical Algorithms, 1, 195--211
(1998)

\bibitem{jour2} Pryce, J.D.: A simple structural analysis method for DAEs. BIT,
41(2), 364--394 (2001)

\bibitem{jour3} Pantelides, C.: The Consistent Initialization of Differential-Algebraic Systems, SIAM J. Sci. and Stat. Comput. 9(2), 213--231
(1988)

\bibitem{jour10} Soares, R.P. and Secchi, A.R.: Direct Initialisation and Solution of High-Index DAE Systems, ESCAPE 15, Barcelona, Spain, 157--162 (2005)

\bibitem{book2} Fritzson, P.: Principles of Object-Oriented Modeling and Simulation with Modelica 2.1, Wiley-IEEE Press
(2004)
\bibitem{jour11} Soares, R.P. and  Secchi, A.R.: EMSO: A New Environment for Modeling, Simulation and Optimization. ESCAPE 13, Lappeenranta, Finlandia, 947--952 (2003)

\bibitem{jour4} Corless, R.M., Ilie, S.: Polynomial cost for solving
IVP for high-index DAE. BIT, 48, 29--49 (2008)

\bibitem{proceeding1} Wu, W.Y., Reid, G.: Symbolic-numeric Computation of Implicit Riquier Bases
for PDE. Proc. of ISSAC'07, ACM, 377--385 (2007)

\bibitem{jour12}
Mani, N.: Fast numeric geometric techniques for computer generated
DAE models. MSc Thesis, University of Western Ontario (2010)

\bibitem{jour6}Leitold, A., Gerzson, M.: Structural analysis of process models
using their reprsentation graph. Hungarian journal of industrial
chemistry veszpr$\acute{e}$m, 37(2), 145--151 (2009)

\bibitem{jour7} Leitold, A., Gerzson, M.: Structural Decomposition of Process Models
Described by Higher Index DAE Systems. Computer Aided Chemical
Engineering, 28, 385--390 (2010)

\bibitem{jour8} Kuhn, H.W.: The Hungarian Method for the assignment problem,
Naval Research Logistics Quarterly, 2, 83--97 (1955)

\bibitem{jour9} Reid, G., Wittkopf, A.D., Boulton A.: Reduction of systems of nonlinear
 partial differential equations to simplified involutive forms. Eur. J. Appl. Math. 7(6), 635--666 (1996)




\end{thebibliography}
\end{document}